\documentclass[10pt,aps,pra,twoside,showpacs,twocolumn,groupedaddress]{revtex4-1}
\usepackage{amssymb,amsmath,amsfonts,amsthm}%authblk
\usepackage{bm}
\usepackage{color}%cite
\usepackage{dsfont}
\usepackage{epsfig}
\usepackage{hyperref}
\usepackage{lipsum}
\usepackage{pbox}
\usepackage{rotating}
\usepackage{stmaryrd}
\newcommand{\be}{\begin{equation}}
\newcommand{\ee}{\end{equation}}
\newcommand{\bea}{\begin{eqnarray}}
\newcommand{\eea}{\end{eqnarray}}
\newcommand{\ket}{\rangle}
\newcommand{\bra}{\langle}

\begin{document}
\newtheorem{theorem}{Theorem}
\newtheorem{proposition}[theorem]{Proposition}
\newtheorem{corollary}[theorem]{Corollary}
\newtheorem{open problem}[theorem]{Open Problem}
\newtheorem{Definition}{Definition}

\title{A local model of quantum Turing machines}

\author{Dong-Sheng Wang}
\email{wdscultan@gmail.com}
\affiliation{Institute for Quantum Computing and Department of Physics and Astronomy,
University of Waterloo,  Waterloo, ON, N2L 3G1, Canada}

%%%%%%%%%%%%%%%%%<<<<<<<<<<<<<<<<<<<<<<<<<<<<<<<<<<<>>>>>>>>>>>>>>>>>>>>>>>>>>>>>>>>>>>>>>>>>>>>>>
%\tableofcontents

\begin{abstract}
The model of local Turing machines is introduced,
including classical and quantum ones,
in the framework of matrix-product states.
The locality refers to the fact that at any instance of the computation
the heads of a Turing machine have definite locations.
The local Turing machines are shown to be equivalent to the corresponding circuit models
and standard models of Turing machines by simulation methods.
This work reveals the fundamental connection between tensor-network states
and information processing.
\end{abstract}

\date{\today}

\maketitle

\section{Introduction}
\label{sec:intr}

Matrix-product states and tensor-network states
have been playing central roles in quantum information science~\cite{AKLT87,FNW92,PVW+07,Sch11}.
They could be used, but not limited,
to characterize entanglement in many-body systems~\cite{ECP10},
construct topological quantum error-correcting codes~\cite{Kit03},
and enable universal quantum computing in the measured-based models~\cite{BBD+09}.

In this work, we reveal a fundamental connection
between matrix-product states and quantum computation.
We find that a universal quantum Turing machine~\cite{Deu85,BV97,Yao93},
which is equivalent to the usual quantum circuit model,
can be defined in the framework of matrix-product states.
On the one hand,
our model of quantum Turing machines greatly simplifies
the functionality of the standard ones
with a locality structure;
on the other hand,
our model establishes a sort of `duality' between
information processing and
matrix-product states (and also tensor-network states),
hence bringing together perspectives and results
from both sides.

Quantum computation is a computing model that operates
according to quantum mechanical rules.
In general, information can be processed by the interaction between registers
(i.e., string of bits) and an external drive,
or interaction among the registers~\cite{KW97,For03,Wat09,AB09};
quantum Turing machines (QTM) belongs to the first case,
and quantum circuit model (QCM) belongs to the second case.
While QCM has been the canonical model for quantum computing,
QTM, as a universal computing model,
are relatively less understood for physicists~\cite{Lan61,Ben73,Llo00,Mye97,Oza98,ON00,Shi02,MO05,FHJ+07,PJ06,LLS09,MW12,MW19}.

Turing machine, which lies at the heart of the theory of computation,
is a universal mathematical or computational model to study algorithms
and the process of computation.
Physically, a Turing machine (TM) is a bipartite system including a register tape and a processor,
interaction between which is enabled by a read/write head.
The standard QTM~\cite{Deu85,BV97,Yao93} is `fully' quantum~\cite{Mye97,Oza98,ON00,Shi02}
in the sense that all elements of TM are quantized,
which results in the superposition of head position
and the nonlocal interactions between the tape and the processor.
Observe that quantum systems do not have to be fully quantum;
on the contrary,
there are usually classical ingredients
(e.g., sites on a lattice, temperature, external potential)
in models of quantum systems.
We find that,
a local and simplified QTM can be defined
by making the head position classical,
without loss of universality for quantum computing.
Local interactions are physically appealing, and
it is not hard to see
that it is present in the circuit models (with local gates),
as well as the classical TM (CTM).

There are also other study that provides hints for variations of TM.
With teleportation and gates realized by nonunitary means~\cite{GC99},
a measurement-based QTM was introduced~\cite{PJ04,PJ06},
which, by construction, does not have problems such as
halt qubit~\cite{Mye97,Oza98,Shi02}
and the locality issue of head position.
With qubits as passive memory, i.e., no direct interactions among them,
computing models with projections on ancilla~\cite{AOK+10}
and automatically decoupled ancilla~\cite{PAK13}
are proven to be universal by simulating a universal set of gates.
In both models each register qubit may be acted upon many times,
i.e., the interaction is not sequential.
In quantum optics, a so-called qubus model~\cite{SNB+06} was developed,
which realizes gates on qubits by non-sequential interactions with a quantum bus,
which is infinite dimensional.

Motivated by the observation above,
the model of local TM is introduced in this work,
including classical and quantum ones.
The probabilistic (or stochastic) TM and quantum stochastic TM are also introduced,
and are shown to be reducible to local CTM and QTM, respectively.
The local structure is brought and imprinted onto a TM
from the matrix-product states (MPS) formalism~\cite{AKLT87,FNW92,PVW+07,Sch11}.

This work contains the following sections.
We first present MPS formalism and also develop techniques that are suitable for TMs in section~\ref{sec:mpsprop}.
We then define local TMs in section~\ref{sec:tm}.
A study of probabilistic TMs is also presented in the Appendix.

\section{Matrix-product states}
\label{sec:mpsprop}

\subsection{Quantum channels}
\label{subsec:qchannel}

We first review some basic properties of quantum channels,
which are needed to understand matrix-product states.
From Stinespring dilation theorem and Kraus operator-sum representation~\cite{Sti55,Kra83,Cho75},
a quantum channel $\mathcal{E}$, i.e.,
a completely positive trace preserving (CPTP) map,
can be represented as
\begin{equation}\label{eq:chano}
  \mathcal{E}(\rho)=\sum_\ell K_\ell \rho K_\ell^\dagger, \; \forall \rho
\end{equation}
for a set of Kraus operators $\{K_\ell\}$
and $\sum_\ell K_\ell^\dagger K_\ell =\mathds{1}$.
Furthermore, the set of Kraus operators corresponds to an isometry operator
$V:=\sum_\ell |\ell\rangle K_\ell$ for $|\ell\rangle$ as ancilla state.
The isometry can be embedded into a unitary operator $U$ such that $V=U|0\rangle$
and $K_\ell=\langle \ell|U|0\rangle$.

The transfer matrix~\cite{Cho75,BZ06} of a quantum channel $\mathcal{E}$ is
\begin{equation}\label{eq:transfer}
\mathcal{T}_\mathcal{E}=\sum_\ell K_\ell\otimes K_\ell^*.
\end{equation}
The dynamics $\mathcal{E}:\rho\mapsto \mathcal{E}(\rho)$
is equivalent to $\mathcal{T}_\mathcal{E}:|\rho\rangle \mapsto \mathcal{T}_\mathcal{E}|\rho\rangle$,
for a quantum state $\rho=\sum_{ij}\rho_{ij}|i\rangle\langle j|$
and $|\rho\rangle=\sum_{ij}\rho_{ij}|i\rangle| j\rangle$.
For $K_\ell=\sum_{ij} k^\ell_{ij}|i\rangle\langle j| $, then
$\mathcal{T}_\mathcal{E}=\sum_{\ell ijkl} k^\ell_{ij} \bar{k}^\ell_{kl} |ik\rangle\langle jl|.$
Ignoring the coherence part the matrix
\begin{equation}\label{}
  \mathcal{S}_\mathcal{E}=\sum_{\ell ij} k^\ell_{ij} \bar{k}^\ell_{ij} |ii\rangle\langle jj|
\end{equation}
is stochastic (not doubly) as $\sum_{\ell i} |k^\ell_{ij}|^2=1$,
and can be treated as the stochastic version of $\mathcal{T}_\mathcal{E}$.
Note $|ii\rangle$ can be simply viewed as
an encoding of $|i\rangle$, same for $\langle jj|$.
For example,
$\mathcal{S}_\mathcal{U}=  \sum_{ ij} |u_{ij}|^2 |i\rangle\langle j|$
of a unitary operator $U$ is doubly stochastic, and also orthostochastic.
The stochastic version of a random unitary channel is also doubly stochastic.

\subsection{Matrix-product states and quantum circuits}
\label{subsec:mpsqc}

Any finite-dimensional $N$-partite quantum state can be expressed as a MPS
\begin{equation}\label{eq:mps0}
  |\Psi\rangle=\sum_{i_1,\dots,i_N} \langle A^{i_N}|A^{i_{N-1}} \cdots A^{i_2} |A^{i_1}\rangle |i_1 \dots i_N\rangle
\end{equation}
for the open boundary condition (OBC) case.
This form is in the so-called right-canonical form,
while left-canonical and other forms are also available~\cite{Sch11}.
These $A$ matrices act on the so-called correlation space,
also known as virtual space, ancillary space etc,
and the correlation space dimension $\chi$
is also known as the bond or virtual dimension.
Regarding LQTM,
the $N$ particles (or 'spins', qubits) are on the tape,
and the correlator is the processor.
Tracing out the system results in a sequence of quantum channels
$\mathcal{E}_n$ on the correlator
such that
$\mathcal{E}_n(\rho)=\sum_{i_n} A^{i_n} \rho A^{i_n \dagger}$,
and $\sum_{i_n}A^{i_n\dagger}A^{i_n}=\mathds{1}$ for each $n=1,\dots,N$.

The boundary condition is specified by the set of column vectors $\{|A^{i_1}\rangle\}$
and the set of row vectors $\{\langle A^{i_N}|\}$.
For the first site, $\sum_{i_1}A^{i_1\dagger}A^{i_1}=1$, each $A^{i_1}$ is a column vector but not-normalized,
while its norm is a singular value.
For the last site, $\sum_{i_N}A^{i_N\dagger}A^{i_N}=\mathds{1}$, each $A^{i_N}$ is a row vector and normalized,
and they come from each column of a unitary operator
that appears in the first step of singular value decomposition (SVD) to derive MPS~\cite{Sch11}.
For the OBC case, the form of MPS is usually simplified as
\begin{align}\label{eq:tiobcmps}
|\Psi\rangle=
\sum_{i_1,\dots,i_N} \langle R| A^{i_N} \cdots A^{i_1} |L\rangle |i_1\dots i_N\rangle,
\end{align}
which may not be normalized due to
the probability of the final projection $\langle R|$.
However, the normalization condition can be easily handled, so it does not cause problem.
For the PBC case, the MPS takes the following form
\begin{align}\label{}
|\Psi\rangle=\sum_{i_1,\dots,i_N} \text{tr}(A^{i_N}\cdots A^{i_1})|i_1\dots i_N\rangle.
\end{align}
We observe that this state can be prepared by using $|\omega\rangle=\sum_i |ii\rangle$ as
both the initial and final states of the correlator.
The bond dimension is actually $\chi^2$,
but the $A$ matrices only act on half of the space,
so the effective bond dimension is still $\chi$.
Also the PBC case can be viewed as a special case of OBC
when each vector $|A^{i_1}\rangle$ is equivalent to $|A^{i_N}\rangle$
and $\{|A^{i_1}\rangle\}$ forms a basis of the correlation space.

Next we study how to prepare a MPS~(\ref{eq:mps0}) by a quantum circuit.
To do so, the dilation for each of the channels $\mathcal{E}_n$ is employed.
The first channel $\mathcal{E}_1$ is defined by the set of Kraus operators $\{|A^{i_1}\rangle\}$,
and the last channel $\mathcal{E}_N$ is defined by $\{\langle A^{i_N}|\}$.
The channel $\mathcal{E}_1$ maps from dimension $\chi_0=1$ to dimension $\chi_1$,
while the channel $\mathcal{E}_N$ maps from dimension $\chi_{N-1}$ to dimension $\chi_N=1$,
while each other channel $\mathcal{E}_n$ in between
maps the bond dimension from $\chi_{n-1}$ to $\chi_n$.
From the SVD process there exists relations between each $\chi$ and $d$,
e.g., $\chi_{N-1}\leq d$~\cite{Sch11}.
Implementing each $\mathcal{E}_n$ requires the dilation of channels that alter dimension.
For a rank-$r$ CPTP channel from dimension $n$ to $m$,
one input ancilla with dimension $\lceil\frac{rm}{n}\rceil$ is needed.
Note that the input system and ancilla do not correspond to the output system and ancilla, respectively,
due to the change of dimension.
Now a channel $\mathcal{E}_n$ can be realized by a unitary $U_n$ with dimension $d \chi_n $,
and from $A^{i_n}=\langle i_n|U_n|0\rangle$,
$\{A^{i_n}\}$ occupy the first block-column of $U_n$.
For the last unitary $U_N$, special cares are needed.
If $\chi_{N-1}=d$, then no ancilla is needed,
which means the correlator itself becomes the last physical spin,
and then it is traced out after a unitary rotation $U_N$ such that $\langle A^{i_N}|=\langle i_N| U_N$,
which appears in the first step of SVD for the right-canonical form.
If $\chi_{N-1}< d$ then an ancilla is needed
and $\langle A^{i_N}|=\langle i_N| U_N|0\rangle$
for $|0\rangle$ as the initial state of this ancilla.

The whole state preparation process is as follows.
First, apply a sequence of unitary gates from $U_1$ till $U_{N-1}$
  \begin{align}\label{}
  &U_{N-1}\cdots U_1|0\rangle_v|0\rangle_1\cdots |0\rangle_{N-1} \\ \nonumber
  &=\sum_{i_1,\dots, i_{N-1}}A^{i_{N-1}}\cdots A^{i_2}|A^{i_1}\rangle|i_1\dots i_{N-1}\rangle,
\end{align}
where $|0\rangle_v$ is the initial state of the virtual correlator.
If $\chi_{N-1}=d$, apply $U_N$ first and trace out the correlator, the state becomes
    \begin{align}\label{}
      &\sum_{i_N}\langle i_N| U_NU_{N-1}\cdots U_1|0\rangle_v|0\rangle_1\cdots |0\rangle_{N-1} |i_N\rangle  \\ \nonumber
      &= \sum_{i_N}\sum_{i_1,\dots, i_{N-1}}\langle i_N| U_NA^{i_{N-1}}\cdots A^{i_2}|A^{i_1}\rangle|i_1\dots i_{N-1}\rangle |i_N\rangle,
    \end{align}
which is the MPS~(\ref{eq:mps0}).
If $\chi_{N-1}<d$, append the last ancilla with $|0\rangle$ such that $W:=U_N|0\rangle$ and $\langle A^{i_{N}}|=\langle i_N|W$.
Applying $U_N$ and tracing out the final system (both correlator and ancilla)
yields the MPS~(\ref{eq:mps0}).

As we can see, the change of bond dimension complicates the MPS circuit,
so instead,
these matrices can be enlarged to have the same bond dimension as the largest one,
and indeed, in practice many states can be described by MPS with constant bond dimensions.
Therefore, it can be assumed that all the $A$ matrices have dimension $\chi$,
and each quantum channel becomes dimension-preserving.
For the quantum circuit, the first dilation $U_1$ maps from dimension $d\chi$ to $d$-dimensional spin and $\chi$-dimensional correlator,
and the channels in the middle are simple to deal with,
while the last one deserves some attention.
The set $\{\langle A^{i_N}|\}$ still forms a channel,
but now it may hold $d\leq \chi$, while injectivity requires $d\geq \chi^2$.
This means for both injective and also $\chi^2\geq d\geq \chi$ cases
the method described above can be used.
For the case $d< \chi$, the channel cannot be TP since each vector $\langle A^{i_N}|$ is extended to a larger vector.
This means partial projection on the correlator is required,
which leads to probabilistic events.
However, we can employ the method in subsection~\ref{subsec:finalp} to avoid this.

\subsection{Avoid the final projection on correlation space}
\label{subsec:finalp}

\begin{figure}
  \centering
  \includegraphics[width=.4\textwidth]{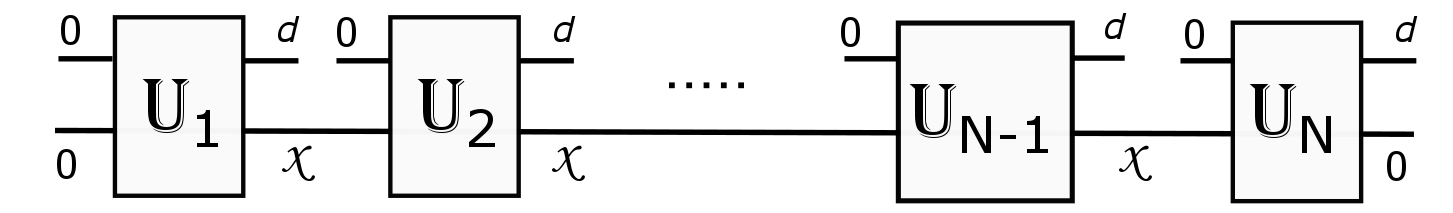}\\
  \caption{Quantum circuit to prepare a general MPS with constant bond dimension $\chi$
  and an automatically decoupled correlator $|0\rangle$ at the end.
  Each unitary $\mathbb{U}_n$ is the dilation of $\mathbb{V}_n$,
  the embedding of $V'_n$.
  The initial state of the correlator can be chosen to be $|0\rangle$
  by absorbing a unitary gate, which converts $|0\rangle$ to $|L'\rangle$,
  into the first gate $\mathbb{U}_1$.
  }\label{fig:mps-same-bond}
\end{figure}

Consider the generation of a MPS in the form~(\ref{eq:tiobcmps})
with a constant bond dimension.
From Ref.~\cite{SSV+05} a MPS can be prepared deterministically such that
the correlator is decoupled at the end,
here this method is extended for the general qudit case.

With the isometry $V_n:=\sum_{i_n}|i_n\rangle A^{i_n}$ for each site,
a MPS with OBC~(\ref{eq:tiobcmps}) can be written as
\begin{equation}\label{}
  |\Psi\rangle=\langle R| V_N\cdots V_1|L\rangle.
\end{equation}
With $\langle R| V_N=(\mathds{1}_d\otimes \langle R| )V_N$ and from SVD
 \begin{equation}\label{}
   (\mathds{1}_d\otimes \langle R| )V_N=V'_NM_N,
 \end{equation}
for (i) $d<\chi$, $M_N$ of size $d\times \chi$, unitary $V'_N$ of size $d\times d$,
and (ii) $d\geq \chi$, $M_N$ of size $\chi\times \chi$, isometry $V'_N$ of size $d\times \chi$.
Now $M_NV_{N-1}$ is $(\mathds{1}_d\otimes M_N )V_{N-1}$, and perform SVD for the rest sites,
and for the last one define $|L'\rangle=M_1|L\rangle$, so
\begin{equation}\label{}
  |\Psi\rangle=V'_N\cdots V'_1|L'\rangle.
\end{equation}
From a rank consideration, the size of $V'_{N-k}$ is $d\min(\chi,d^k)\times \min(\chi,d^{k+1})$,
and the size of $M_k$ is always at most $\chi\times \chi$.
Now each $V'$ can be embedded into an isometry $\mathbb{V}$ of size $d\chi\times \chi$,
although the embedding is not unique.
This means a quantum circuit to realize the sequence of $\mathbb{V}_k$ can be used to prepare the MPS:
start from the state $|L'\rangle$, and perform the dilation $\mathbb{U}_k$ for each $\mathbb{V}_k$.
To show that the correlator can automatically decouple at the end,
there are three cases to consider:
\begin{enumerate}
  \item For $d^2\geq\chi>d$,
  the size of $V'_{N-1}$ is $d^2\times \chi$,
  while the size of its embedding $\mathbb{V}_{N-1}$ is $d\chi\times \chi$.
This embedding can be done by appending $\chi-d$ rows of zeros to each of the $d\times \chi$ matrices in $V'_{N-1}$,
and this means after the action of $V'_{N-1}$, the $\chi$-level correlator will only have amplitude on $d$ levels.
The embedding $\mathbb{V}_N$ can be obtained by first appending $\chi-d$ columns of normalized vectors,
and then inserting $\chi-1$ rows of zeros after each row in $V'_N$,
and this means that the state of the correlator will be annihilated by $V'_N$, i.e., mapped to dimension one,
and the correlator is converted to the last spin by $V'_N$.
  \item For $\chi\leq d$, the size of $V'_{N-1}$ is $d\chi\times \chi$, and its embedding is the same with itself;
and size of $V'_N$ is $d\times \chi$, and its embedding $\mathbb{V}_N$
can be obtained by inserting $\chi-1$ rows of zeros after each row in $V'_N$.
In this case, after $V'_{N-1}$ all levels of the correlator are occupied,
yet $V'_N$ will still annihilate the correlator.
  \item For $\chi>d^2$, the size of $V'_{N-1}$ is $d^2\times d^2$, and its embedding can be obtained
by first appending $\chi-d^2$ columns of normalized vectors,
and then appending $\chi-d$ rows of zeros to each of the $d\times \chi$ matrices in $V'_{N-1}$.
Still in this case after $V'_{N-1}$ only $d$ levels of the correlator are occupied,
which are further annihilated by $\mathbb{V}_N$.
\end{enumerate}
The quantum circuit can be shown as that in Fig.~\ref{fig:mps-same-bond}.
This also shows that in a LQTM one does not need to implement a projection on the correlator,
which is the processor of LQTM.

Here we apply this technique to the Bell states
$|\Phi^\pm\ket=\frac{1}{\sqrt{2}} (|00\ket\pm|11\ket)$,
and $|\Psi^\pm\ket=\frac{1}{\sqrt{2}} (|01\ket\pm|10\ket)$.
We only need the MPS form for $|\Phi^+\ket\equiv |\omega\ket$, and others can be easily obtained.
Let the two qubits be $\alpha$ and $\beta$, and a qubit ancilla be $a$,
we find
\be |\omega\ket=\bra 0|_a B A |0\ket_a,\ee
for $A=|0\ket_\alpha A^0 + |1\ket_\alpha A^1$,
$B=|0\ket_\beta B^0 + |1\ket_\beta B^1$,
with the tensors defined as
\be A^0=\mathds{1}/\sqrt{2},\; A^1=\sigma^x/\sqrt{2},\; B^0=P_0,\; B^1=\sigma^+,\ee
for $P_0=(\mathds{1}+\sigma^z)/2$, $\sigma^+=(\sigma^x+i\sigma^y)/2$,
and Pauli matrices $\sigma^x$, $\sigma^y$, $\sigma^z$,
and $|0\ket=(1,0)^t$, $|1\ket=(0,1)^t$.
The pair of matrices $A^0$ and $A^1$, $B^0$ and $B^1$ each form a quantum channel.
The quantum circuit to prepare $|\omega\ket$ is also easy to find
\be |\omega\ket= \bra 0|_a U_{\beta a} U_{\alpha a} |000\ket_{\beta \alpha a}, \ee
for $U_{\beta a}=S_{\beta a}$ as a swap gate realizing $B^0$ and $B^1$,
$U_{\alpha a}=\text{CNOT}_{\alpha a} H_\alpha$ with
the controlled-not ($\alpha$ as control) and Hadamard gate realizing $A^0$ and $A^1$.
The qubit ancilla $a$ automatically decouples simply because it is swapped with the qubit $\beta$.

\subsection{Composition}
\label{subsec:comp}

In the MPS circuit the starting state of system is usually
$|{\bf 0}\rangle\equiv|0\cdots 0\rangle$.
If the input $|\bf{0}\rangle$ is substituted by another MPS,
the output is still a MPS, but with a larger bond dimension.
Such a composition is useful when we consider a sequence of computations by a LQTM.

Let's denote a MPS by $|\Xi_a\rangle$
and the sequence of unitary operators in it as $\mathcal{U}^{(a)}$,
and $|\Xi_a\rangle := \langle R_a | \mathcal{U}^{(a)} |L_a\rangle |\bf{0}\rangle$ with bond dimension $\chi_a$,
and similarly for another MPS by $|\Xi_b\rangle$.
The composition of the two circuits leads to the state
\begin{align}\label{}
|\Xi_{ab}\rangle =  \langle R_b | \mathcal{U}^{(b)} |L_b\rangle |\Xi_a\rangle
=  \langle R_b |\langle R_a | \mathcal{U}^{(ab)}  |L_b\rangle |L_a\rangle |\bf{0}\rangle,
\end{align}
with $\mathcal{U}^{(ab)}:= \mathcal{U}^{(a)} \diamond \mathcal{U}^{(b)} $
for composition $\diamond$ defined as follows.
For $\mathcal{U}^{(a)}:=\prod_i U_i^{(a)}$, $\mathcal{U}^{(b)}:=\prod_i U_i^{(b)}$,
Let $\tilde{U}_i^{(a)}=U_i^{(a)}\otimes \mathds{1}^{(b)}$, $\tilde{U}_i^{(b)}=U_i^{(b)}\otimes \mathds{1}^{(a)}$,
then $\mathcal{U}^{(ab)}=\prod_i U_i^{(ab)}$ for $U_i^{(ab)}:=\tilde{U}_i^{(a)}\tilde{U}_i^{(b)}$.
The state $|\Xi_{ab}\rangle$ has bond dimension $\chi_{ab}=\chi_a \chi_b$,
and the boundary states of the correlator are
$|L_b\rangle |L_a\rangle$ and $\langle R_b |\langle R_a|$.
This property also holds when the technique to avoid the final projection
from section~\ref{subsec:finalp} is employed.

In addition, the tensor product of $|\Xi_{a}\rangle $ and $|\Xi_{b}\rangle $
also yields a new MPS $|\Xi_{a\otimes b}\rangle$ with
bond dimension $\chi_{a\otimes b}=\chi_a \chi_b$ and
\begin{align}\label{}
|\Xi_{a\otimes b}\rangle
= & \langle R_b |\langle R_a | \mathcal{U}^{(a\otimes b)}  |L_b\rangle |L_a\rangle |\bf{0}\bf{0}\rangle,
\end{align}
with $\mathcal{U}^{(a\otimes b)}=\prod_i U_i^{(a\otimes b)}$ for $U_i^{(a\otimes b)}:=U_i^{(a)}\otimes U_i^{(b)}$.

\section{Turing machines}
\label{sec:tm}

\subsection{Preliminary}
\label{subsec:pretm}

We first review the standard description of TM~\cite{BV97,For03}.
For convenience, we will use `cbit' for classical bit,
and `bit' as a general notion for a cit, pbit, or qubit.
A TM, classical, probabilistic, or quantum,
has a processor (also known as control),
denoted by the symbol~$Q$, state of which is often called `internal state',
and a register (tape) $\Gamma$ of a string of non-interacting bits,
which usually contains the input and output,
and a head, which can read, write, move left or right by at most one step, see Fig.~\ref{fig:tm}.
Usually the processor $Q$ is specified to have an initial internal state $q_0\in Q$
and a set of halting states $F\subseteq Q$
so that the machine halts when the internal state reaches a halting state.
There is a transition function $\delta$ which forms the program to solve a certain problem.
The transition function takes the form
\begin{equation}\label{eq:tmtransition}
\delta: Q \backslash F \times \Gamma \times Q \times \Gamma \times \{L,R,N\}\rightarrow
\mathcal{D},
\end{equation}
for
$  \mathcal{D}_\text{CTM}=\{0,1\},\; \mathcal{D}_\text{PTM}=[0,1],\; \mathcal{D}_\text{QTM}=\mathbb{C}$~\cite{BV97,For03}.
Here $L$ (left), $R$ (right), and $N$ (no movement) specifies the motion of the head~\cite{puncture}.

\begin{figure}[t!]
  \centering
  \includegraphics[width=.3\textwidth]{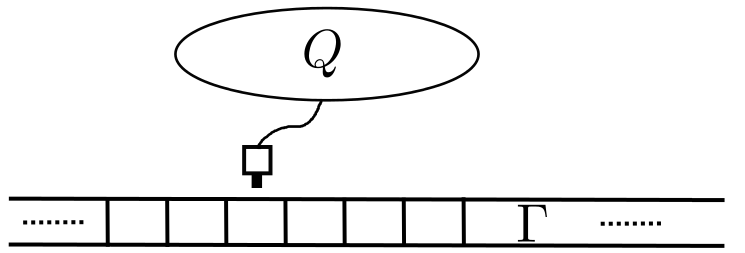}\\
  \caption{The model of Turing machines.}
  \label{fig:tm}
\end{figure}

A state of the whole machine is often known as a `configuration',
including the state of tape, head position, processor (and some others).
A computation on TM can be viewed as a sequence of configurations,
and a conversion between any successive two configurations can be described
by a permutation, stochastic process, or unitary evolution.
The tape is a {\em passive} memory, i.e.,
the bits on the tape do not interact with each other.
This means the computation is not carried out on the tape itself,
instead it is induced by the interaction between the tape and the processor.

To define a QTM~\cite{Deu85,BV97},
we employ the tuple form $\langle Q, F, H, P, \Gamma, \Sigma, \delta\rangle$:
\begin{enumerate}
  \item $Q$:  Hilbert space of the internal states.
  \item $F\subseteq Q$: the set of starting and halting states of the processor.
  \item $H$: Hilbert space of the halt qubit.
  \item $P$: Hilbert space of head position $P=\{|p\rangle\}$.
  \item $\Gamma$: Hilbert space of the quantum tape.
  \item $\Sigma \subseteq \Gamma$: the set of input states on the tape.
  \item $\delta_\text{QTM}$: transition map $Q \times \Gamma \times P \times H \rightarrow Q \times \Gamma \times P \times H $.
\end{enumerate}
Compared with the CTM and PTM, it is clear that
all the components in the configuration of a machine
including the head position, tape, processor, transition map and the halt operation
become quantum.

The tape is formed by a string of non-interacting qubits,
while each qubit has a quantum position index $|p\rangle$.
The transition map $\delta_\text{QTM}$ corresponds to a unitary operator,
and there are two important features of it.
First, the state of the head position $P$
has quantum correlations with the tape $\Gamma$ and the processor $Q$,
so the evolution on the tape and processor itself would not be unitary
if the computing part is isolated from the global unitary on the whole configuration.
Second, the `quantum walk' of the head can at most shift one position
in each step,
i.e., from $|p\rangle$ to $|p\pm 1, 0\rangle$.
This can be understood as a kind of locality in the space $P$,
which, however, does not correspond to a locality in the real space,
which is still a classical space.
After several steps of computation,
there will be a superposition of the head position and one will not be able to see where the head sits,
and the interaction between the processor $Q$ and the tape $\Gamma$ will become nonlocal.
The head also has quantum correlations with $\Gamma$ and $Q$,
so the evolution on $\Gamma$ and $Q$ would not be unitary
if the other parts are traced out.

\subsection{Local Turing machines}
\label{subsec:ltm}

The form~(\ref{eq:tmtransition}) is a \emph{global} description of a TM
and does not reveal the locality of interactions explicitly.
For a CTM, the head has a definite position in each step,
and needs to move in both directions to achieve universality.
The interaction between the processor and one cbit of the tape
is two-body and local in the space of position,
i.e., the `real' space
(compared to the momentum space) in physics.
A PTM can be viewed as a randomized CTM,
and the computation by a PTM is a randomized permutation,
i.e., a (doubly) stochastic process.
Each `trajectory' of PTM is a CTM,
hence the local structure of interactions in CTM carries over to PTM (Appendix~\ref{sec:ptm}).

To make the physical locality explicit and simplify
the functionality of TM,
we now introduce the model of local TM (LTM).
We will prove its universality and draw the connection with
matrix-product states (MPS).
Instead of a global description, a LTM is described via
the local interactions between the tape $\Gamma$ and the processor $Q$.
We will show that the model of LTM is equivalent to the standard TM
and the circuit model correspondingly.
A LTM is specified as follows.
\begin{Definition}
A LTM is represented by a tuple
$\langle Q, F, \Gamma, \Sigma, \delta_c, \delta_s, C\rangle$:
\begin{enumerate}
  \item $Q$: Space of the internal states.
  \item $F\subseteq Q$: the set of starting and halting states of the processor.
  \item $\Gamma$: Space of the tape as a product of local ones, $\Gamma_n$.
  \item $\Sigma \subseteq \Gamma $: the set of input states on the tape.
  \item $\delta_c$: local computing map
  $Q \times \Gamma_n \rightarrow Q \times \Gamma_n$.
  \item $\delta_s$: classical head position shift function $\mathbb{Z}\rightarrow \mathbb{Z} : p_\ell \mapsto p_{\ell+1}$ for $p_{\ell+1}=p_\ell\pm1,0$.
  \item $C$: classical control, i.e., a finite set of classical internal states.
\end{enumerate}
\end{Definition}

The local space $\Gamma_n$ on the tape is that for a bit.
The processor $Q$ can be represented as a set of bits,
which could interact with each other or not,
while the tape $\Gamma$
contains a string of non-interacting bits.
The sets $F$ and $\Sigma$ are defined for completeness,
yet we will not explicitly analyze their roles in this work.

The computing maps (or gates) $\delta_c$ specify a two-body interaction
between the processor and each bit on the tape.
Four types of gate are possible:
permutation, stochastic process, unitary evolution,
or completely positive trace preserving (CPTP) map,
also known as quantum channel or quantum stochastic process~\cite{NC00}.
The unitary (channel) case generalizes the permutation (stochastic) case.
%Due to dilation, quantum channels are equivalent to unitary gates (on enlarged systems).
There might be final measurements on the tape for the LQTM,
as we will see later on.

The classical control $C$ is
formed by a set of classical states $\{c\}$ that corresponds to the computing part,
and it has a starting state $c_0$ for the starting state,
and some halting states $\{c_f\}$ for the halting states of the processor.
The function of $C$ is to signal the process of the machine such that
the machine halts when the classical control is at a halting state.
It is also implicitly present in QCM while usually not mentioned.

The fundamental way to prove universality and study the relation among various models
is by simulation~\cite{AB09}.
There are many kinds of simulations according to convergence of variables or operator topology~\cite{Deu85,BV97,Yao93,Shi02,Nest11,BJS11,Wan15,FG13}.
Our framework of simulation is as follows.
The simulation of a TM $\textsc{m}$ by another TM $\textsc{u}$
is a task such that
\begin{equation}\label{}
\textsc{u}([\textsc{m}],[x])=[\textsc{m}(x)],\; \forall x,
\end{equation}
here $[\cdot]$ represents encoding, e.g.,
$[\textsc{m}]$ is the bit-string description of $\textsc{m}$.
The simulation is efficient if
there is only a polynomial overhead of cost for all input $x$.
Furthermore, as $[\textsc{m}]$ is only being read during the simulation,
$[\textsc{m}]$ does not have to be the input of $\textsc{u}$,
hence in fact
$\textsc{u}([x])=[\textsc{m}(x)]$, $\forall x$,
and there exists a program
\begin{equation}\label{}
  \textsc{p}([\textsc{m}],[x])=[\textsc{u}],\; \forall x,
\end{equation}
such that $\textsc{p}$ specifies
the process of $\textsc{u}$ to simulate $\textsc{m}$ on arbitrary $x$.
Each $x$ is an input of $\textsc{p}$ since
the simulation is to simulate the action of $\textsc{m}$ on $x$,
and both $\textsc{p}$ and $\textsc{u}$ are generically $x$-independent.
We will focus on simulation efficiency without a specification of simulation accuracy,
which simplifies our study and does not affect our conclusions.

\begin{proposition}
The models of LCTM, CTM, and CCM are equivalent.
\end{proposition}
\begin{proof}
We only need to show the equivalence between LCTM and CTM,
LCTM and CCM, since CTM and CCM are known to be equivalent.
Given a computation on LCTM,
with a processor $Q$ and a tape $\Gamma$ of a certain size,
each permutation $\Pi$ acts on a tape bit and the processor $Q$.
The simulation by a CTM is simple
by observing that each step in CTM is a permutation $\Pi$.
A gate $\Pi$ can be simulated by a sequence of Boolean gates in CCM.
A Boolean circuit acts on $|\Gamma|+|Q|$ bits can simulate the LCTM efficiently,
for $|Q|$ ($|\Gamma|$) as the number of bits to represent states of $Q$ ($\Gamma$).

Given a CTM which is a sequence of configurations,
the simulation by LCTM is as follows.
If at step $\ell$ the head position is $p_\ell$,
the symbol at position $p_\ell$ on the tape is $\gamma(p_\ell)$,
and the internal state is $q_\ell\in Q \backslash F$,
then the transition to the next step
is simulated by a shift operation on the head $p_\ell \mapsto p_{\ell+1}$,
for $p_{\ell+1}=p_\ell\pm1,0 \in \mathbb{Z}$,
and a permutation operation on the corresponding tape bit and the processor to realize
$(q_\ell, \gamma(p_\ell)) \mapsto (q_{\ell+1}, \gamma(p_{\ell+1}))$.
Given a Boolean circuit, each gate in it can be simulated efficiently by a
local permutation in a LCTM.
\end{proof}

It is also clear to see there exists a \emph{universal} LTM such that
it can simulate a given LTM efficiently.
The universality can also be seen from the universality of Boolean circuits,
which states that any Boolean function can be computed efficiently by a Boolean circuit.

When the bits in LTM are qubits and the computing maps are unitary operations,
we arrive at a LQTM.
The key difference from the classical cases is the quantum superposition,
which is an additional feature and shall not be viewed as a generalization of mixing.
Mixing and probability can be included in the quantum formalism by quantum channels.
However, replacing unitary operations by quantum channels do not change
the computational power of quantum computers~\cite{AKN97}
due to the dilation theorem~\cite{Sti55,Kra83,Cho75}.
As far as we know,
the dilation theorem does not exist for the classical case;
namely, it is not clear if a stochastic process can be embedded in a permutation on a larger space.
Instead, a doubly stochastic process is a convex sum of permutations on the same space.

It is well known that QCM is equivalent to QTM~\cite{BV97,Yao93,For03,MW19},
so we will not show the equivalence between QTM and LQTM directly.
Instead, we will show the equivalence between QCM and LQTM.

\begin{proposition}
The models of LQTM and QCM are equivalent.
\end{proposition}
\begin{proof}
Given a unitary circuit $U$ in the QCM,
its gates are assumed from the universal gate set $\{CZ,H,T\}$~\cite{NC00}.
The Hadamard gate $H$ and $T$ gate can be easily simulated.
Each gate $CZ_{ij}$ acting on qubits $i$ and $j$
can be simulated easily with a qubit ancilla $e$ at state $|0\ket$
which belongs to the processor with
\be CZ_{ij} |\psi_i\ket|\psi_j\ket |0\ket_e
=S_{ie}CZ_{je}S_{ie}|\psi_i\ket|\psi_j\ket |0\ket_e \label{eq:cz}\ee
for swap gate $S$.
As a result, the circuit $U$ can be efficiently simulated by a LQTM.
The simulation of a LQTM $\textsc{m}$ by a circuit in QCM
is simple: with the states of the processor encoded by qubits,
each local gate in $\textsc{m}$ can be simulated by an array of gates,
and in all simulated efficiently by a quantum circuit.
\end{proof}

The universality of QCM transfers to LQTM.
We observe that the simulation of the gate CZ~(\ref{eq:cz}) is non-sequential.
It is known that direct sequential unitary simulation of entangling gates are impossible~\cite{LLP+08}.
However, with teleportation we find the structure of LQTM can be further simplified.

\begin{proposition}
There exists a unilateral universal LQTM.
\end{proposition}
\begin{proof}
The non-sequential simulation of the gate CZ~(\ref{eq:cz}) can be converted as a sequential one
with teleportation gadget.
With the MPS form of the Bell state $|\omega\ket$ discussed in section~\ref{subsec:finalp},
a gate $CZ_{ij}$ can be simulated as
\bea &&CZ_{ij}|\psi_i\ket|\psi_j\ket|0\ket_e |0\ket_a \\ \nonumber
&&=\sigma^m_\alpha M_{i\beta}^m
U_{\beta a}CZ_{je}S_{\alpha e}U_{\alpha a}S_{ie}
|\psi_i\ket |0\ket_\alpha |\psi_j\ket |0\ket_\beta |0\ket_e |0\ket_a. \eea
The Bell measurement $M_{i\beta}^m$ with Pauli correction $\sigma^m_\alpha$
will teleport the state of $i$ to $\alpha$.
This product of unitary gates is sequential,
and qubits $e$ and $a$ belong to the processor,
and qubits $i$ and $\beta$ on the tape will be measured.
Such a LQTM is unilateral while the exception is that Bell measurements have to be done
on the tape at the end of the computation.
A Bell measurement here can be simplified to projective measurements on $i$ and $\beta$ since
the state of $i$ is $|0\ket$.
\end{proof}

This compares to the classical case.
A CTM with a one-direction moving head is not universal.
It is also known as a finite state transducer,
which is a deterministic finite automata that the input is only read once~\cite{AB09}.
This highlights the crucial role of entanglement and teleportation for quantum computing.

A unilateral LQTM is nothing but the process to prepare MPS.
The interaction between processor $Q$ and tape $\Gamma$ is sequential,
and furthermore, we showed that $Q$ can be automatically decoupled
at the end of the computation,
with the output contained solely on the tape $\Gamma$.

\subsection{Multipartite setting}
\label{subsec:multi}

\begin{figure*}[t!]
  \centering
  \includegraphics[width=.8\textwidth]{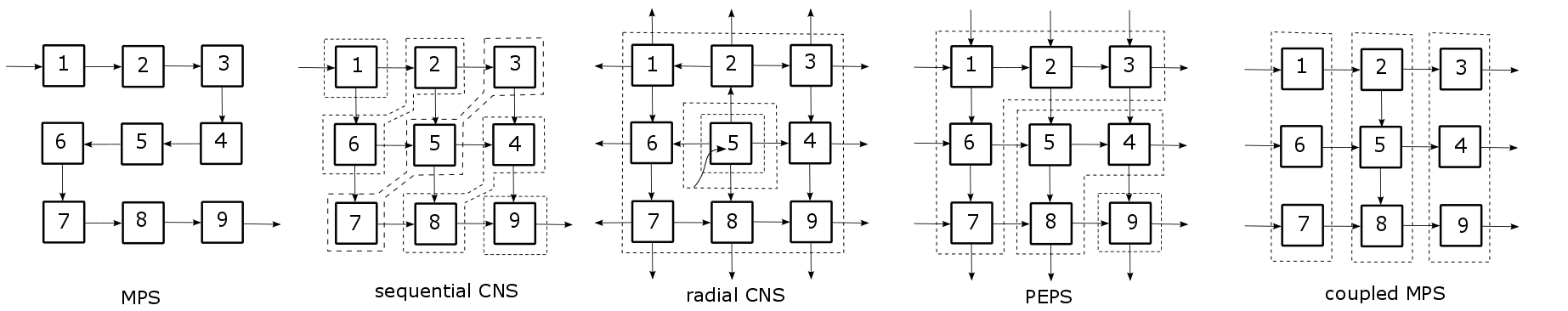}\\
  \caption{Schematic diagrams of examples of TNS.
  The boxes with numbers represent channels,
  after dilation, each channel yields a spin that belongs to the final TNS.
  The left one with a linear information flow is the usual (linear) MPS with nine stages of the flow.
  The dashed circles highlight the stages in other states:
  the sequential TNS has five stages,
  radial TNS has two stages,
  PEPS has three stages,
  and the coupled MPS also has three stages.
  The channels within a stage can also be ordered according to their flows.
  }\label{fig:peps}
\end{figure*}

In general, a LQTM can have multiple tapes and processors as in the classical cases
for various practical purposes,
e.g., each processor can be a small system, even a single qubit.
This requires a slight extension of MPS to tensor-network states (TNS).

In a MPS the correlator (processor)
is acted upon by a sequence of quantum channels,
without a detailed structure of the channels and the free propagations between them.
To make this clear, let us define MPS with more general information flows
as quantum tensor-network states (TNS),
and the standard setting with a linear information flow as (linear) MPS.

Below we introduce TNS from the viewpoint of channel networks.
For a directed acyclic graph $G=(E,V)$ with edge set $E$ and vertex set $V$,
assign a quantum channel to each vertex, and then
the composition of quantum channels
$\mathcal{E}:=\mathcal{E}_{|V|}\circ \cdots \mathcal{E}_2 \circ \mathcal{E}_1$
forms an acyclic quantum channel network.
The requirement of cycle-free is to avoid causality problem,
i.e., the output of a channel cannot become the input of it at a later time.
Also note here in this section,
a general CP map, which may not be trace-preserving or dimension-preserving,
is viewed as a quantum channel.

The output of $\mathcal{E}_\ell$ is from the input of $\mathcal{E}_{\ell+1}$,
and each channel $\mathcal{E}_\ell$ can contain several parts, e.g.,
$\mathcal{E}_\ell=\mathcal{E}_{\ell_1}\otimes \mathcal{E}_{\ell_2}$
or $\mathcal{E}_\ell=p_1 \mathcal{E}_{\ell_1}+ p_2 \mathcal{E}_{\ell_2}$
for $p_1+p_2=1$ as a convex combination,
or other complicated forms,
and accordingly, the input and output of each channel $\mathcal{E}_\ell$
can contain several parts.
Similar definitions can be found in other settings~\cite{Wat18,CDP09,KW05},
and have been called quantum networks~\cite{CDP09},
quantum channels with memory~\cite{KW05},
while the definition above is used to introduce TNS.

Given the Kraus operator representation $\{K_i\}$ of a channel $\mathcal{E}$,
which may not be trace-preserving or dimension-preserving,
the channel can also be written as an operator $V=\sum_i |i\rangle K_i$,
which is an isometry if trace-preserving.
For a channel network $\mathcal{E}=\mathcal{E}_N\circ \cdots \mathcal{E}_2 \circ \mathcal{E}_1$
with boundary states $|I\rangle$ and $\langle O|$,
given the operator $V_\ell$ of each channel $\mathcal{E}_\ell$,
a TNS is
\begin{equation}\label{}
|\Psi\rangle= \langle O| \prod_\ell V_\ell |I\rangle.
\end{equation}

The channels act on the correlation space.
Examples of TNS are shown in Fig.~\ref{fig:peps}.
The linear MPS has the simplest information flow structure,
while all others are still MPS but with branches.
Note that the motivation to allow non trace-preserving or dimension-preserving channels
is that the norm of a TNS does not play a central role.
For instance, in the PEPS form the channel at each vertex is usually not trace-preserving,
and in the coupled MPS form there are dimension-altering channels.
In many-body physics, MPS is usually used to represent 1D systems,
while PEPS is used for 2D or 3D systems on different lattices.

The flow in a TNS represents the evolution of the ancilla (correlator),
corresponding to the sequence of matrix multiplications (or tensor contractions).
For LQTM, the lattice of qubits form the tape,
the correlator serves as the processor,
and there might be qubits that need to be measured at the end of computation
due to the simulation of entangling gates.
A multipartite LQTM has a more complicated information flow structure,
with its output described as a tensor-network state (TNS),
which can still be simulated by a single-tape single-processor universal LQTM,
since a TNS is also a MPS with a larger bond dimension.

\begin{proposition}
  A LQTM with multiple tapes and/or processors
  can be simulated by a single-tape single-processor universal LQTM.
\end{proposition}
\begin{proof}
If the processor is single-partite while the tape is $m$-partite,
then in the MPS circuit each unitary operator acts on the processor
and $m$ qubits, with one from each tape.
If the processor is also multi-partite,
then this leads to multi-tape multi-processor machine,
for which there could be coupling between different parts,
and this corresponds to the coupled MPS scheme shown in Fig.~\ref{fig:peps}.
The resulting TNS on the tape is still a MPS,
which can be prepared on a universal LQTM.
\end{proof}

\section{Conclusion}
\label{sec:conc}

In this work a model of local Turing machines is introduced.
We show that the model of local classical (quantum) Turing machine
is equivalent to the model of standard classical (quantum) Turing machine
and classical (quantum) circuit model.
The structure of a
local quantum Turing machine can be described based on matrix product states
(and teleportation).
Our work simplify the construction of quantum Turing machines
and establish a close relationship with quantum many-body systems.

While the interaction between a tape bit and the processor
seems no more easier than that between bits,
the model is suitable for situations when direct interaction among bits is difficult,
such as distributed computing and communication
and when the tape and processor encoded in different physical systems.
Models like the qubus model in quantum optics can be viewed as special kinds of
local Turing machines.
Finally, the processor can also contain multiple parts,
and the design and complexity of its structure are nontrivial subject on its own.

This work does not intend to study Turing machine from the viewpoint of computer scientists.
Issues like grammar, language, complexity etc,
and relations with other universal quantum computing models
shall also be pursued for separate investigations.

\section{Acknowledgements}

This research has been supported by NSERC.
The author acknowledges the Department of Physics and Astronomy
at the University of British Columbia, Vancouver,
where the first draft of the paper was completed.

\appendix

\section{Probabilistic Turing machines}
\label{sec:ptm}

Here we present a study of probabilistic Turing machines (PTM) and the local versions.
Recall that a PTM can be understood as a randomized CTM,
and the randomness can be realized by random variables,
which can be encoded by a string of pbits on a so-called random tape,
and the computation by a PTM is a randomized permutation,
which can be described by a doubly stochastic matrix.
Also each step of a PTM is a stochastic matrix
$S=\sum_\lambda \mathfrak{p}_\lambda \Pi_\lambda$
for a set of permutations $\Pi_\lambda$ with probability $\mathfrak{p}_\lambda$,
which is represented on the random tape.
The product of a sequence of stochastic matrices can be expressed as
\begin{equation}\label{}
\prod_i S_i=\sum_{\lambda_1,\lambda_2,\dots} \mathfrak{p}_{\lambda_1}\mathfrak{p}_{\lambda_2}\cdots \left(\prod_i  \Pi_{\lambda_i}\right),
\end{equation}
and each sequence in the parenthese above represents a CTM with corresponding probability.
That is, a CTM realizes a particular trajectory of a PTM.
Each permutation $\Pi_{\lambda_i}$ acts on the processor and a single cbit.
The output of a PTM contains the final states $\gamma\in \Gamma$ on the tape with probability
\begin{equation}\label{eq:p}
  \mathcal{P}(\gamma)= \sum_{p\in \mathbb{Z},q\in Q} \mathcal{P}(p, q, \gamma),
\end{equation}
where the sum is over position $p$ and internal state $q$
for the same $\gamma$.

Observe that the PTM is \emph{fully} probabilistic:
the computation on the whole configuration of the machine is stochastic.
As a result, there is also a probability distribution of the head position:
it is uncertain where the head is during each step of the computation.
However, this actually does not cause physical problems
thanks to different interpretations of probability:
the frequency interpretation and ensemble interpretation.
In the former one, probability is the ratio $n/N$ of
the number of times $n$ for the occurrence of a particular event
to the total amount of runs $N$.
In the latter one,
given a total amount $N$ of a collection of objects,
the probability of a particular object is the weight $n/N$ given
$n$ copies of this object.
The probability in PTM is in the frequency interpretation.
As a result, a PTM can be viewed as a randomized CTM.
However, there is no such interpretations of quantum superposition,
which causes the subtlety of locality for QTM as we studied in the main text.

A LPTM can be defined by deleting the randomness of head positions.
A LPTM can also be viewed as a LCTM with one additional tape of pbits.
Given a LPTM, it can be simulated by a PTM easily
since LPTM is a restricted version of PTM.
Given a PTM, the simulation by a LPTM contains two steps:
first decompose the PTM as pbits and a collection of CTMs,
then the CTMs each can be simulated by a LCTM according to the pbits.
With this, we see that the model of LPTM is equivalent to PTM,
except that the pbits are given as a free resource.

\subsection{Stochastic matrix product states}
\label{sec:smps}

Here we show that pbits (when not free) can be prepared by a LPTM.
The reason is that,
the states of pbits, as probability vectors, can be written as stochastic MPS (sMPS)(see, e.g., Ref.~\cite{TV10}).
We show that each pbit on the tape is only acted upon once,
i.e., the read/write head is unilateral,
and the processor is automatically decoupled at the end.

It is shown that~\cite{TV10} any probability vector $|p\rangle$ can be written as
a sMPS form
\begin{equation}\label{}
  |p\rangle=\sum_{i_1,\dots, i_N} A_{i_1}^{[1]} P^{[1]} A_{i_2}^{[2]}\cdots P^{[N-1]} A_{i_N}^{[N]}
  |i_1\dots i_N\rangle
\end{equation}
such that $S^{[n]}:=P^{[n-1]}C^{[n]}$ is a stochastic matrix for
$  C^{[n]}=\sum_{i_n} A_{i_n}^{[n]}$.
Furthermore, we find this can also be proved using the non-negative matrix factorization (NMF) method~\cite{LS99,LS01,HD08,YZY+11}.
A matrix is non-negative iff all its entries are equal to or greater than zero.
In particular, given a $m\times n$ non-negative matrix $A$,
it can be well approximated by
\begin{equation}\label{}
  A'=PDQ^t
\end{equation}
such that the generalized Kullback-Leibler divergence
$D(A||A')$ is minimized for $k\leq \min(m,n)$,
wherein $P$ is $m\times k$, $Q$ is $n\times k$, and both are column stochastic,
and $D$ is diagonal non-negative such that $\sum_i D_{ii}=\sum_{ij} A_{ij}$~\cite{HD08}.
The elements $D_{ii}$ play similar roles with singular values.

Given a multi-partite probability vector $|p\rangle$ written as
$|p\rangle=\sum_{i_1,\dots, i_N}^d p({i_1,\dots, i_N}) |i_1\dots i_N\rangle$,
define a matrix $\mathfrak{C}$ with dimension $d\times d^{N-1}$ and elements $\mathfrak{C}_{i_1,(i_2,\dots,i_N)}=p(i_1,\dots,i_N)$.
By NMF $\mathfrak{C}=PDQ^{t}$ and
\begin{equation}\label{}
\mathfrak{C}_{i_1,(i_2,\dots,i_N)}=
\sum_{a_1}^{r_1} P_{i_1, a_1} D_{a_1, a_1}Q^{t}_{a_1,(i_2,\dots,i_N)},
\end{equation}
for $r_1\leq d$.
Denote $D_{a_1, a_1}Q^{t}_{a_1,(i_2,\dots,i_N)}=p(a_1,i_2,\dots,i_N)$,
and a \emph{row} vector $B^{i_1}$ with element $B^{i_1}_{a_1}=P_{i_1, a_1}$, then $\mathfrak{C}_{i_1,(i_2,\dots,i_N)}=\sum_{a_1}^{r_1}B^{i_1}_{a_1}p(a_1,i_2,\dots,i_N)$.
Put $B^{i_1}$ on the most left.
The coefficients $p(a_1,i_2,\dots,i_N)$ can form a new matrix $\mathfrak{C}'$.
By NMF again
\begin{equation}\label{}
\mathfrak{C}_{i_1,(i_2,\dots,i_N)}
=\sum_{a_1}^{r_1}\sum_{a_2}^{r_2}B^{i_1}_{a_1}B^{i_2}_{a_1,a_2}p(a_2,i_3,\dots,i_N),
\end{equation}
for $r_2\leq r_1 d$,
and elements $B^{i_2}_{a_1,a_2}$ form a $r_1\times r_2$ matrix.
At the end
\begin{equation}\label{}
p(i_1,\dots,i_N)=\sum_{a_1,\dots,a_N}^{r_1,\dots,r_N} B^{i_1}_{a_1} B^{i_2}_{a_1,a_2}\cdots B^{i_{N-1}}_{a_{N-2},a_{N-1}} B^{i_N}_{a_N-1},
\end{equation}
and also
\begin{align}\label{}
|p\rangle=\sum_{i_1,\dots,i_N}^d \langle B_{i_1}^{[1]}| B_{i_2}^{[2]}\cdots
B_{i_{N-1}}^{[N-1]} |B_{i_N}^{[N]}\rangle |i_1\dots i_N\rangle
\end{align}
such that each $S^{[n]}:=\sum_{i_n} B_{i_n}^{[n]}$ is column stochastic.
The dimension of $B$ matrices is upper bounded by $d^{N/2-1}\times d^{N/2}$.
Note this is a left-canonical form,
a right-canonical form and mixed form can also be derived
analog with the quantum case~\cite{Sch11}.
Also two boundary probability vectors $\langle \ell|$ and $|r\rangle$ can be pulled out
such that
\begin{equation}\label{}
  |p\rangle=\sum_{i_1,\dots, i_N}\langle \ell| B_{i_1}^{[1]} \cdots B_{i_N}^{[N]} |r\rangle
  |i_1\dots i_N\rangle.
\end{equation}

The next problem now is to automatically decouple the correlator
from the system at the final step.
The method is to apply NMF sequentially again.
Let $S_n=\sum_{i_n} B_{i_n}^{[n]} |i_n\rangle$.
Now assume the bond dimension is $\chi$.
First, as the matrix $(\mathds{1}\otimes \langle \ell|) S_N$ is non-negative,
it can be factorized as
\begin{equation}\label{}
  (\mathds{1}\otimes \langle \ell|) S_N=S_N' T_N
\end{equation}
for $S_N'$ column stochastic and $T_N$ non-negative.
The matrix $T_N S_{N-1}$ can be factorized again,
and then
\begin{equation}\label{}
  |p\rangle=S_N' \cdots S_1' |r'\rangle,
\end{equation}
for each $S_n'$ column stochastic and a probability vector $|r'\rangle$.
Now each $S_n'$ can be embedded into a column stochastic matrix $\mathbb{S}_n$
of size $d\chi \times \chi$ and as the result,
\begin{equation}\label{}
  |p\rangle=\mathbb{S}_N \cdots \mathbb{S}_2  \mathbb{S}_1 |r'\rangle.
\end{equation}
Given $\mathbb{S}_n=\sum_{i_n} |i_n\rangle B_{i_n}^{[n]}$,
a non-unique square column-stochastic matrix $\mathbb{Q}_n$
of dimension $d\chi$ can be defined such that
$\mathbb{S}_n$ occupies its first block-column.
The matrix $\mathbb{Q}_n$ can be viewed as $\mathcal{S}_{\mathcal{E}_n}$,
the stochastic version of $\mathcal{T}_{\mathcal{E}_n}$
for a quantum channel $\mathcal{E}_n$ according to section~\ref{subsec:qchannel}.
As the result, any probability vector can be generated sequentially
using stochastic matrices $\{\mathbb{Q}_n\}$,
each acting on the correlator and a pbit initialized at $|0\rangle$,
such that the correlator is automatically decoupled at the end.

%\newpage
\bibliography{ext}{}

\begin{thebibliography}{10}

\bibitem{AKLT87}
Ian Affleck, Tom Kennedy, Elliott~H Lieb, and Hal Tasaki.
\newblock Rigorous results on valence-bond ground states in antiferromagnets.
\newblock {\em Phys. Rev. Lett.}, 59(7):799, 1987.

\bibitem{FNW92}
Mark Fannes, Bruno Nachtergaele, and Reinhard~F Werner.
\newblock Finitely correlated states on quantum spin chains.
\newblock {\em Commun. Math. Phys.}, 144(3):443--490, 1992.

\bibitem{PVW+07}
David Perez-Garc{\'\i}a, Frank Verstraete, Michael~M Wolf, and J~Ignacio Cirac.
\newblock Matrix product state representations.
\newblock {\em Quant. Inform. Comput.}, 7(5-6):401--430, 2007.

\bibitem{Sch11}
Ulrich Schollw{\"o}ck.
\newblock The density-matrix renormalization group in the age of matrix product
  states.
\newblock {\em Ann. Phys.}, 326(1):96--192, 2011.

\bibitem{ECP10}
J.~Eisert, M.~Cramer, and M.~B. Plenio.
\newblock Colloquium: Area laws for the entanglement entropy.
\newblock {\em Rev. Mod. Phys.}, 82:277--306, Feb 2010.

\bibitem{Kit03}
A~Yu Kitaev.
\newblock Fault-tolerant quantum computation by anyons.
\newblock {\em Annals of Physics}, 303(1):2--30, 2003.

\bibitem{BBD+09}
Hans~J Briegel, David~E Browne, W~D{\"u}r, Robert Raussendorf, and Maarten
  Van~den Nest.
\newblock Measurement-based quantum computation.
\newblock {\em Nat. Phys.}, 5(1):19--26, 2009.

\bibitem{Deu85}
David Deutsch.
\newblock Quantum theory, the {C}hurch-{T}uring principle and the universal
  quantum computer.
\newblock In {\em Proceedings of the Royal Society of London A: Mathematical,
  Physical and Engineering Sciences}, volume 400, pages 97--117. The Royal
  Society, 1985.

\bibitem{BV97}
Ethan Bernstein and Umesh Vazirani.
\newblock Quantum complexity theory.
\newblock {\em SIAM Journal on Computing}, 26(5):1411--1473, 1997.

\bibitem{Yao93}
A.~Chi-Chih Yao.
\newblock Quantum circuit complexity.
\newblock In {\em Foundations of Computer Science, 1993. Proceedings., 34th
  Annual Symposium on}, pages 352--361. IEEE, 1993.

\bibitem{KW97}
Attila Kondacs and John Watrous.
\newblock On the power of quantum finite state automata.
\newblock In {\em Proceedings 38th Annual Symposium on Foundations of Computer
  Science}, pages 66--75. IEEE, 1997.

\bibitem{For03}
Lance Fortnow.
\newblock One complexity theorist's view of quantum computing.
\newblock {\em Theoretical Computer Science}, 292(3):597--610, 2003.

\bibitem{Wat09}
John Watrous.
\newblock Quantum computational complexity.
\newblock {\em Encyclopedia of complexity and systems science}, pages
  7174--7201, 2009.

\bibitem{AB09}
Sanjeev Arora and Boaz Barak.
\newblock {\em Computational complexity: a modern approach}.
\newblock Cambridge University Press, Cambridge U.K., 2009.

\bibitem{Lan61}
Rolf Landauer.
\newblock Irreversibility and heat generation in the computing process.
\newblock {\em IBM J. Res. Dev.}, 5(3):183--191, 1961.

\bibitem{Ben73}
C.~H. Bennett.
\newblock Logical reversibility of computation.
\newblock {\em IBM J. Res. Dev.}, 17(6):525--532, 1973.

\bibitem{Llo00}
Seth Lloyd.
\newblock Ultimate physical limits to computation.
\newblock {\em Nature}, 406(6799):1047--1054, 2000.

\bibitem{Mye97}
John~M. Myers.
\newblock Can a universal quantum computer be fully quantum?
\newblock {\em Phys. Rev. Lett.}, 78:1823--1824, Mar 1997.

\bibitem{Oza98}
Masanao Ozawa.
\newblock Quantum nondemolition monitoring of universal quantum computers.
\newblock {\em Phys. Rev. Lett.}, 80:631--634, Jan 1998.

\bibitem{ON00}
Masanao Ozawa and Harumichi Nishimura.
\newblock Local transition functions of quantum {T}uring machines.
\newblock {\em RAIRO-Theoretical Informatics and Applications},
  34(05):379--402, 2000.

\bibitem{Shi02}
Yu~Shi.
\newblock Remarks on universal quantum computer.
\newblock {\em Phys. Lett. A}, 293(5):277--282, 2002.

\bibitem{MO05}
Takayuki Miyadera and Masanori Ohya.
\newblock On halting process of quantum {T}uring machine.
\newblock {\em Open Systems \& Information Dynamics}, 12(03):261--264, 2005.

\bibitem{FHJ+07}
Willem Fouch{\'e}, Johannes Heidema, Glyn Jones, and Petrus~H Potgieter.
\newblock {Deutsch's universal quantum Turing machine (Revisited)}, 2007.
\newblock arXiv:quant-ph/0701108.

\bibitem{PJ06}
Simon Perdrix and Philippe Jorrand.
\newblock Classically-controlled quantum computation.
\newblock {\em Electronic Notes in Theoretical Computer Science},
  135(3):119--128, 2006.

\bibitem{LLS09}
Antonio~A. Lagana, M.~A. Lohe, and Lorenz von Smekal.
\newblock Construction of a universal quantum computer.
\newblock {\em Phys. Rev. A}, 79:052322, May 2009.

\bibitem{MW12}
Dieter~van Melkebeek and Thomas Watson.
\newblock Time-space efficient simulations of quantum computations.
\newblock {\em Theory of Computing}, 8(1):1--51, 2012.

\bibitem{MW19}
Abel Molina and John Watrous.
\newblock Revisiting the simulation of quantum turing machines by quantum
  circuits.
\newblock {\em Proc. Royal Soc. A}, 475(2226):20180767, 2019.

\bibitem{GC99}
Daniel Gottesman and Isaac~L Chuang.
\newblock Demonstrating the viability of universal quantum computation using
  teleportation and single-qubit operations.
\newblock {\em Nature}, 402(6760):390--393, 1999.

\bibitem{PJ04}
Simon Perdrix and Philippe Jorrand.
\newblock {Measurement-based quantum {T}uring machines and their universality},
  2004.
\newblock arXiv:quant-ph/0404146.

\bibitem{AOK+10}
Janet Anders, Daniel K.~L. Oi, Elham Kashefi, Dan~E. Browne, and Erika
  Andersson.
\newblock Ancilla-driven universal quantum computation.
\newblock {\em Phys. Rev. A}, 82:020301, Aug 2010.

\bibitem{PAK13}
Timothy~J. Proctor, Erika Andersson, and Viv Kendon.
\newblock Universal quantum computation by the unitary control of ancilla
  qubits and using a fixed ancilla-register interaction.
\newblock {\em Phys. Rev. A}, 88:042330, Oct 2013.

\bibitem{SNB+06}
Timothy~P Spiller, Kae Nemoto, Samuel~L Braunstein, William~J Munro, Peter van
  Loock, and Gerard~J Milburn.
\newblock Quantum computation by communication.
\newblock {\em New J. Phys.}, 8(2):30, 2006.

\bibitem{Sti55}
W.~Forrest Stinespring.
\newblock {Positive Functions on C*-algebras}.
\newblock {\em Proc. Am. Math. Soc.}, 6(2):211--216, Apr 1955.

\bibitem{Kra83}
Karl Kraus.
\newblock {\em States, Effects, and Operations: Fundamental Notions of Quantum
  Theory}, volume 190 of {\em Lecture Notes in Physics}.
\newblock Springer-Verlag, Berlin, 1983.

\bibitem{Cho75}
Man-Duen Choi.
\newblock Positive linear maps on complex matrices.
\newblock {\em Linear Algebra Appl.}, 290(10):285--290, 1975.

\bibitem{BZ06}
Ingemar Bengtsson and Karol \.{Z}yczkowski.
\newblock {\em Geometry of Quantum States}.
\newblock Cambridge University Press, Cambridge U.K., 2006.

\bibitem{SSV+05}
C.~Sch\"on, E.~Solano, F.~Verstraete, J.~I. Cirac, and M.~M. Wolf.
\newblock Sequential generation of entangled multiqubit states.
\newblock {\em Phys. Rev. Lett.}, 95:110503, Sep 2005.

\bibitem{puncture}
Also $\mathbb{C}$ needs to be substituted by a punctured version $\mathbb{C}'$
  according to an efficiency argument~\cite{BV97,For03}, and there are also
  blank symbols (vacuum) on the tape, while here only physical issues are
  concerned.

\bibitem{NC00}
Michael~A. Nielsen and Isaac~L. Chuang.
\newblock {\em Quantum Computation and Quantum Information}.
\newblock Cambridge University Press, Cambridge U.K., 2000.

\bibitem{Nest11}
Maarten Van~den Nest.
\newblock Simulating quantum computers with probabilistic methods.
\newblock {\em Quant. Inform. Comput.}, 11(9-10):784–--812, 2011.

\bibitem{BJS11}
Michael~J. Bremner, Richard Jozsa, and Dan~J. Shepherd.
\newblock Classical simulation of commuting quantum computations implies
  collapse of the polynomial hierarchy.
\newblock In {\em Proceedings of the Royal Society A: Mathematical, Physical
  and Engineering Science}, pages 459--472, 2011.

\bibitem{Wan15}
Dong-Sheng Wang.
\newblock Weak, strong, and uniform quantum simulations.
\newblock {\em Phys. Rev. A}, 91:012334, Jan 2015.

\bibitem{FG13}
Bert~E Fristedt and Lawrence~F Gray.
\newblock {\em A modern approach to probability theory}.
\newblock Springer Science \& Business Media, 2013.

\bibitem{AKN97}
Dorit Aharonov, Alexei Kitaev, and Noam Nisan.
\newblock Quantum circuits with mixed states.
\newblock In {\em Proceedings of the Thirtieth Annual ACM Symposium on Theory
  of Computation (STOC)}, pages 20--30. ACM, 1997.

\bibitem{LLP+08}
L.~Lamata, J.~Le\'on, D.~P\'erez-Garc\'{\i}a, D.~Salgado, and E.~Solano.
\newblock Sequential implementation of global quantum operations.
\newblock {\em Phys. Rev. Lett.}, 101:180506, Oct 2008.

\bibitem{Wat18}
John Watrous.
\newblock {\em The theory of quantum information}.
\newblock Cambridge University Press, 2018.

\bibitem{CDP09}
Giulio Chiribella, Giacomo~Mauro D'Ariano, and Paolo Perinotti.
\newblock Theoretical framework for quantum networks.
\newblock {\em Phys. Rev. A}, 80:022339, Aug 2009.

\bibitem{KW05}
Dennis Kretschmann and Reinhard~F. Werner.
\newblock Quantum channels with memory.
\newblock {\em Phys. Rev. A}, 72:062323, Dec 2005.

\bibitem{TV10}
Kristan Temme and Frank Verstraete.
\newblock Stochastic matrix product states.
\newblock {\em Phys. Rev. Lett.}, 104:210502, May 2010.

\bibitem{LS99}
Daniel~D. Lee and H.~Sebastian Seung.
\newblock Learning the parts of objects by non-negative matrix factorization.
\newblock {\em Nature}, 401(6755):788--791, 1999.

\bibitem{LS01}
Daniel~D. Lee and H.~Sebastian Seung.
\newblock Algorithms for non-negative matrix factorization.
\newblock In {\em Advances in neural information processing systems}, pages
  556--562, 2001.

\bibitem{HD08}
Ngoc-Diep Ho and Paul Van~Dooren.
\newblock Non-negative matrix factorization with fixed row and column sums.
\newblock {\em Linear Algebra Appl.}, 429(5):1020--1025, 2008.

\bibitem{YZY+11}
Zhirong Yang, He~Zhang, Zhijian Yuan, and Erkki Oja.
\newblock Kullback-leibler divergence for nonnegative matrix factorization.
\newblock In {\em International Conference on Artificial Neural Networks},
  pages 250--257. Springer, 2011.

\end{thebibliography}
\bibliographystyle{unsrt}

\end{document}